\documentclass [twoside,reqno,12pt] {amsart}

\usepackage{color}
\usepackage{amsfonts}
\usepackage{amssymb}
\usepackage{a4}

\usepackage[toc,page]{appendix}

\newtheorem{thm}{Theorem}[section]

\newtheorem{prop}[thm]{Proposition}

\theoremstyle{definition}

\theoremstyle{remark}

\numberwithin{equation}{section}

\renewcommand{\Re}{\hbox{Re}\,}
\renewcommand{\Im}{\hbox{Im}\,}

\newcommand{\C}{\mathbb{C}}

\newcommand{\N}{\mathbb{N}}

\newcommand{\R}{\mathbb{R}}

\newcommand{\supp}{\operatorname{supp}}

\def\hat{\widehat}
\def\tilde{\widetilde}
\def \bfo {\begin {eqnarray*} }
\def \efo {\end {eqnarray*} }
\def \ba {\begin {eqnarray*} }
\def \ea {\end {eqnarray*} }
\def \beq {\begin {eqnarray}}
\def \eeq {\end {eqnarray}}
\def \supp {\hbox{supp }}

\def \det {\hbox{det}}

\def \p {\partial}

\def\hat{\widehat}
\def\tilde{\widetilde}
\def \bfo {\begin {eqnarray*} }
\def \efo {\end {eqnarray*} }
\def \ba {\begin {eqnarray*} }
\def \ea {\end {eqnarray*} }
\def \beq {\begin {eqnarray}}
\def \eeq {\end {eqnarray}}
\def \supp {\hbox{supp }}

\def \det {\hbox{det}}

\def \p {\partial}

\begin{document}

 \title[Bounds on transmission eigenvalues]{The interior transmission problem and bounds on transmission eigenvalues}

\author[Hitrik]{Michael Hitrik}

\address
        {M. Hitrik,  Department of Mathematics\\
    UCLA\\
    Los Angeles\\
    CA 90095-1555\\
    USA }

\email{hitrik@math.ucla.edu}

\author[Krupchyk]{Katsiaryna Krupchyk}

\address
        {K. Krupchyk, Department of Mathematics and Statistics \\
         University of Helsinki\\
         P.O. Box 68 \\
         FI-00014   Helsinki\\
         Finland}

\email{katya.krupchyk@helsinki.fi}

\author[Ola]{Petri Ola}

\address
        {P. Ola, Department of Mathematics and Statistics \\
         University of Helsinki\\
         P.O. Box 68 \\
         FI-00014   Helsinki\\
         Finland}

\email{Petri.Ola@helsinki.fi}

\author[P\"aiv\"arinta]{Lassi P\"aiv\"arinta}

\address
        {L. P\"aiv\"arinta, Department of Mathematics and Statistics \\
         University of Helsinki\\
         P.O. Box 68 \\
         FI-00014   Helsinki\\
         Finland}

\email{Lassi.Paivarinta@helsinki.fi}

\begin{abstract}

We study the interior transmission eigenvalue problem for sign-definite multiplicative perturbations of the Laplacian in a bounded domain. We show that all but finitely many complex transmission eigenvalues are confined to a parabolic neighborhood of the positive real axis.

\end{abstract}

\maketitle

\section{Introduction and statement of results}

Recently, there has been a large number of  new developments in the study of transmission eigenvalues and the interior transmission eigenvalue problem for elliptic operators with constant coefficients, see e.g.  \cite{CakColHous10, ColMonk88, ColPaiSyl07, HitKruOlaPai, HitKruOlaPai2,  paisyl08}. Transmission eigenvalues play an essential role in reconstruction algorithms of inverse scattering theory in an inhomogeneous medium, such as the sampling method and  the factorization method \cite{CakColbook, ColKir96, KirGribook},  and also carry information about the scatterer \cite{CakColGint_complex, paisyl08}.

The discreteness of the set of transmission eigenvalues was established in  \cite{vanhatkonnatI} in the case of the Laplacian -- see also \cite{HitKruOlaPai} for more general operators. As for the existence of transmission eigenvalues, the first results are due to  \cite{paisyl08}, and the existence of infinitely many real transmission eigenvalues was shown in     \cite{CakDroHou}. Going into the complex spectral plane, the existence of transmission eigenvalues off the real axis has been demonstrated in the recent paper \cite{CakColGint_complex} in a particular situation.

The purpose of this note is to study the location of transmission eigenvalues in the complex plane. We show that the transmission eigenvalues are confined to a parabolic neighborhood of the positive real axis.  To the best of our knowledge, the only previous result concerning the location of transmission eigenvalues is due to  \cite{CakColGint_complex}, where it is proved that under suitable additional assumptions, the transmission eigenvalues belong to the right half plane.

We shall now proceed to recall the precise statement of the interior transmission problem.
Let  $\Omega\subset \R^n$ be a bounded domain with $C^\infty$-boundary $\p \Omega$,  and
 $m\in C^\infty(\overline{\Omega},\R)$  with $m>0$ in $\overline{\Omega}$.  In the context of scattering theory, the function $1+m$ represents the index of refraction  of an inhomogeneous medium, with $\overline{\Omega}$ being the support of the perturbation $m$.

  The interior transmission eigenvalue problem for the operator  $P_0=-\Delta$ is the following degenerate boundary value problem,
\begin{equation}
\label{eq_ITP}
\begin{aligned}
(P_0-\lambda)v=0 \quad &\text{in} \quad \Omega,\\
(P_0-\lambda(1 + m))w=0 \quad &\text{in} \quad \Omega,\\
 v-w \in H^{2}_0(\Omega).
\end{aligned}
\end{equation}
Here
\[
H^2_0(\Omega)=\{u\in H^2(\R^n):\supp (u)\subset \overline{\Omega}\},
\]
where $H^2(\R^n)$ is the standard Sobolev space.

We say that $0\ne\lambda\in \C$ is a transmission eigenvalue if the problem \eqref{eq_ITP} has non-trivial solutions $0\ne v\in L^2(\Omega)$ and  $0\ne w\in L^2(\Omega)$.

The following is the main result of this note.

\begin{thm}
\label{thm_main}
There exist $0<\delta<1$ and $C>1$ that such all
 transmission eigenvalues $\lambda\in \C$ with $|\lambda|>C$
 satisfy
\[
\emph{\Re} \lambda>0,\quad  |\emph{\Im} \lambda|\le C|\lambda|^{1-\delta}.
\]
\end{thm}

\textbf{Remark}. It follows from the proof that we can take $\delta=1/25$.

\begin{figure}[htbp]
\begin{center}
\setlength{\unitlength}{1mm}
\begin{picture}(100,70)
  \put(0,30){\vector(1,0){80}}
    \put(25,0){\vector(0,1){60}}
     \put(40,20){\line(0,1){20}}
     \qbezier(40,20)(45,5)(80,0)
     \qbezier(40,40)(45,55)(80,60)
      \put(25, 63){\makebox(0, 0){$\Im \lambda$}}
       \put(83, 27){\makebox(0, 0){$\Re \lambda$}}
        \put(70, 64){\makebox(0, 0){$\Im \lambda=C|\lambda|^{1-\delta}$}}
         \color{red}
  \linethickness{0.3mm}
  \put(50, 20){\line(1, 0){4}}
  \put(52, 18){\line(0, 1){4}}
  \put(50, 40){\line(1, 0){4}}
  \put(52, 38){\line(0, 1){4}}
  \put(44, 30){\line(1, 0){4}}
  \put(46, 28){\line(0, 1){4}}
  \put(58, 24){\line(1, 0){4}}
  \put(60, 22){\line(0, 1){4}}
  \put(58, 36){\line(1, 0){4}}
  \put(60, 34){\line(0, 1){4}}
  \put(66, 6){\line(1, 0){4}}
  \put(68, 4){\line(0, 1){4}}
  \put(66, 54){\line(1, 0){4}}
  \put(68, 52){\line(0, 1){4}}
  \put(70, 30){\line(1, 0){4}}
  \put(72, 28){\line(0, 1){4}}
   \end{picture}
\caption{All but finitely many transmission eigenvalues are located in a parabolic neighborhood about the positive real axis}
\end{center}
\end{figure}
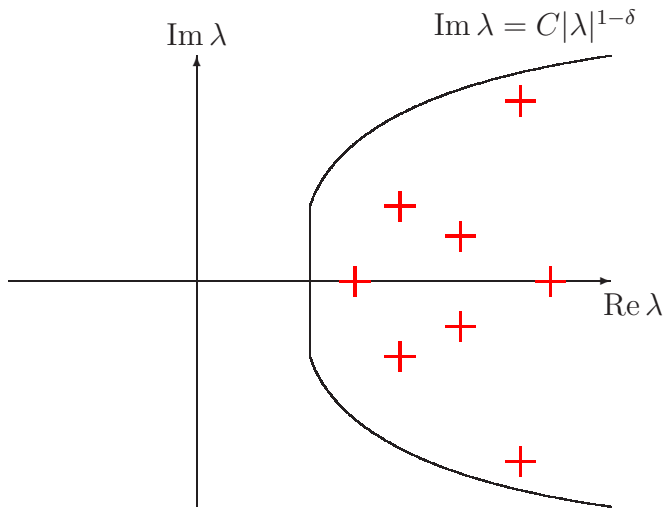

The proof of Theorem \ref{thm_main} is carried out in several  steps. First,  following  \cite{paisyl08}, in Section 2 we reformulate the interior transmission problem \eqref{eq_ITP} as an elliptic boundary value problem for a quadratic operator pencil.  We are interested in the invertibility properties of the pencil in question. It appears that available results on quadratic pencils in the literature such as e.g. \cite{DenkMenVol2000},  \cite{Markus} do not seem to be applicable in our situation. We shall therefore adopt a direct approach, based on methods of the semiclassical analysis. The second step in the proof is a reduction to a semiclassical boundary value problem, given in Section 3.
This problem is  inverted  asymptotically in Section 4, which leads to the absence of transmission eigenvalues in a parabolic neighborhood of the real axis. The final step of the proof of Theorem \ref{thm_main} is done in Section 5, where we show that the left-half plane contains at most finitely many transmission eigenvalues.  In the appendix we have collected some basic facts concerning the semiclassical calculus which are used in the main part of the paper.

It would be interesting to study the distribution of transmission eigenvalues inside of the parabolic region of Theorem \ref{thm_main}. We hope to return to this problem in the future, where the methods of this work could be expected to be applicable.

\section{Reduction to an elliptic boundary value problem}

From \cite{paisyl08} let us recall the following characterization of transmission eigenvalues.
A complex number
$\lambda\ne 0$ is a transmission eigenvalue if and only if there exists $0\ne u\in  H^{2}_0(\Omega)$ satisfying
\[
T(\lambda) u:=(P_0-\lambda(1+m))\frac{1}{m}(P_0-\lambda)u=0.
\]
Notice that by elliptic regularity, $u\in C^\infty(\overline{\Omega})$.

We have
\[
T(\lambda) =A-\lambda B +\lambda^2 C,
\]
where
\[
A=P_0qP_0,\quad
B=qP_0+P_0q+P_0,\quad
C=1+q,\quad
q=\frac{1}{m}.
\]

Let us consider the following boundary value problem,
\begin{equation}
\label{eq_bvp_elliptic_1}
\begin{aligned}
T(\lambda) u=f, \quad &\text{in}\quad \Omega,\\
\gamma_0 u=g_1, \quad &\text{on}\quad \p\Omega,\\
\gamma_0 \p_\nu u=g_2, \quad &\text{on}\quad \p\Omega,
\end{aligned}
\end{equation}
where $\nu$ is the exterior unit normal to the boundary $\p�\Omega$, and $\gamma_0$ is the operator of the restriction to $\p \Omega$.
Let
\[
\mathcal{T}(\lambda):u\mapsto (T(\lambda)u, \gamma_0 u, \gamma_0 \p_\nu u)
\]
 and
\begin{equation}
\label{eq_space_H_s}
\mathcal{H}^s=H^{s-4}(\Omega)\times H^{s-1/2}(\p\Omega)\times H^{s-3/2}(\p \Omega), \quad s> 3/2,
\end{equation}
where
\[
H^t(\Omega)=\{u|_{\Omega}: u\in H^t(\R^n)\},\quad t\in \R,
\]
and $H^t(\p \Omega)$ is the standard Sobolev space on $\p \Omega$.
It is then known that for any $\lambda\in \C$,  \eqref{eq_bvp_elliptic_1} is an elliptic boundary value problem in the classical sense, and hence, the operator
\[
\mathcal{T}(\lambda):H^s(\Omega)\to \mathcal{H}^s, \quad s> 3/2,
\]
is Fredholm, see for instance \cite{ChazPir_book, Grubbbook2009, Taylor_book_1}. In what follows in \eqref{eq_space_H_s} we shall take $s=4$.

\begin{prop}
\label{prop_index_0}

For any $\lambda\in \C$,
$
\emph{ind}(\mathcal{T}(\lambda))=0.
$
\end{prop}

\begin{proof} In
 \cite{HitKruOlaPai} it was shown that the operator $T(0)$, equipped with the domain $H^4(\Omega)\cap H^2_0(\Omega)$ is selfadjoint and  positive.  It follows that  the operator $\mathcal{T}(0)$ is injective.

To see the surjectivity of $\mathcal{T}(0)$ it suffices to notice that the trace operator
\[
(\gamma_0,\gamma_0 \p _\nu): H^4(\Omega)\to H^{4-1/2}(\p\Omega)\times H^{4-3/2}(\p \Omega),
\]
is surjective, as well as $T(0)$.
 Thus,  $\mathcal{T}(0)$ is an isomorphism,  and, hence,  $\mathcal{T}(\lambda)$  has index zero, for each $\lambda\in \C$ .

\end{proof}


\section{Semiclassical reduction}

Let us extend $q\in C^\infty(\overline{\Omega})$ to the whole of  $\R^n$ in such a way that the extension, still denoted by $q$,
satisfies $q\in C^\infty_b(\R^n)$, $q>0$, and  $q$ is a positive constant near infinity. Here
\[
C^\infty_b(\R^n)=\{ u\in C^\infty(\R^n):\p^\alpha u\in L^\infty(\R^n), \forall \alpha\}.
\]
Then  $T(\lambda)$ becomes  an elliptic partial differential operator of order four on $\R^n$, with coefficients in $C^\infty_b(\R^n)$, depending polynomially on $\lambda$.

We shall study the family of operators $T(\lambda)$ in the regime $|\lambda|\gg 1$.  It will be convenient to make a semiclassical reduction of $T(\lambda)$, so that  we write
\[
\lambda=\frac{z}{h^2},
\]
 where $0<h\ll 1$ is a semiclassical parameter and $z\in \C$,
$|z|\sim 1$.
The idea in the semiclassical approach is to write $T=h^4 T(\lambda)$ in the form, where all the partial derivatives $\p_{x_i}$ are multiplied by the semiclassical parameter $h$. In this way we arrive at
\begin{equation}
\label{eq_family_T}
T=T(x,hD_x, z;h)=h^4T(\lambda) =A_h-z B_h +z^2 C_h, \quad \text{in}\quad \R^n,
\end{equation}
\[
A_h=h^2P_0qh^2P_0,\quad
B_h=qh^2P_0+h^2P_0q+h^2P_0,\quad
C_h=1+q.
\]

Let us consider the semiclassical version of the boundary value problem \eqref{eq_bvp_elliptic_1},
\begin{equation}
\label{eq_bvp}
\begin{aligned}
T(x,hD_x, z;h) u&=f, \quad \textrm{in}\quad  \Omega,\\
u|_{\p \Omega}&=g_1, \\
 h D_\nu u|_{\p \Omega}&=g_2,
\end{aligned}
\end{equation}
with
$D_\nu=i^{-1}\p_\nu$.

We have $\overline{T(x,hD_x, z;h) u}=T(x,hD_x,\overline{z};h)\overline{u}$, and therefore, it will suffice to consider the region $\Im z>0$.
The main step in the proof of  Theorem \ref{thm_main} is a construction of  a right parametrix for the boundary value problem \eqref{eq_bvp} in the region $\Im z\ge h^{\delta/2}$, for $\delta >0$ sufficiently small. The semiclassical parametrix  construction implies the existence of a right inverse for the operator
\[
\mathcal{T}: H^4(\Omega)\to \mathcal{H}^4, \quad  \mathcal{T}u=(Tu, \gamma_0 u, \gamma_0 h D_\nu u),
\]
for $\Im z\ge h^{\delta/2}$ and all $h$ small enough. Here the spaces $H^4(\Omega)$ and $\mathcal{H}^4$ are equipped with the natural semiclassical norms.
In view of Proposition \ref{prop_index_0}, this leads to the  absence of transmission eigenvalues in the region $|\lambda|\ge C$ and $|\Im \lambda|\ge C|\lambda|^{1-\frac{\delta}{4}}$,  for some constant $C>0$.
The parametrix  construction for the boundary value problem \eqref{eq_bvp} is carried out in Section \ref{sec_parametrix}  and in Section \ref{sec_left_half} the proof of Theorem \ref{thm_main}  is completed by observing that the left-half plane $\Re \lambda<0$ contains at most finitely many transmission eigenvalues.

\section{Parametrix construction for the problem \eqref{eq_bvp}}
\label{sec_parametrix}

\subsection{Inverting the family $T(x,hD_x,z;h)$ in $\R^n$  }

We shall be concerned with the family $T$ in the region of the complex spectral plane, where  $\Im z\ge h^{\delta/2}$, $\delta>0$ small enough. We refer to the appendix for the notation and basic facts of the calculus of semiclassical pseudodifferential operators.

Let $t=t_0+ht_1$ be the full symbol of $T(x,hD_{x}, z;h)\in \textrm{Op}_h(S^4)$. Here $t_0$ is
the semiclassical leading symbol of $T(x,hD_{x}, z;h)$ given by
\begin{equation}
\label{eq_t_0_factor}
\begin{aligned}
t_0(x,\xi,z)&=q(x)p_0^2(x,\xi)-z(2q(x)+1)p_0(x,\xi)+z^2(q(x)+1)\\
&=q(x)(z-p_0(x,\xi))\left(\frac{q(x)+1}{q(x)}z-p_0(x,\xi)\right),
\end{aligned}
\end{equation}
where $p_0(x,\xi)=\xi^2$, and $t_1\in S^3$.

Since $|z|$ is in a  bounded set and $q,1/q\in L^\infty(\R^n)$, we have
\[
|t_0(x,\xi,z)|\ge \begin{cases} (\Im z)^2, & (x,\xi)\in \R^n\times \R^n,\\
\langle \xi \rangle^{4}/C, & |\xi| \ge C,
\end{cases}
\]
where $C$ is large enough.

We have the following result giving a parametrix construction for $T$ in $\R^n$.

\begin{prop}
\label{prop_interior_parametrix}
Consider the region  $\emph{\Im} z\ge h^{\delta/2}$, $0\le \delta< 1/2$.  Then there exist $r_j\in  S^{\delta+2\delta j, -4-j}_{\delta}$, $j=0,1,\dots$, such that for any $N\in \N$,
\[
T\emph{\textrm{Op}}_h(\sum_{j=0}^N h^j r_j)=I +h^{N+1}\emph{\textrm{Op}}_h(c_N), \quad c_N\in S^{2\delta(N+1),-N-1}_\delta.
\]
Here $r_0=1/t_0$ and $r_j$, $j\ge 1$, are of the form $f_j/g_j$, where $g_j$ is a positive power of $t_0$ and $f_j$ is a polynomial in $z,\xi$, whose coefficients are smooth in $x$.
\end{prop}

\begin{proof}

Set
\[
r_0(x,\xi,z)=\frac{1}{t_0(x,\xi,z)}.
\]
Let us first show that $r_0\in S^{\delta,-4}_\delta$. Indeed,
using the Fa\`a di Bruno formula \cite{Lerner_book} and the fact that $|t_0|\ge h^\delta$, we get, for bounded $|\xi|$,
\[
|\p ^\alpha_x\partial^\beta_\xi r_0|\le C_{\alpha,\beta}h^{-\delta}h^{-\delta(|\alpha|+|\beta|)}.
\]
Since the estimate for large $|\xi|$ is clear, the claim follows.

We have
\[
T\textrm{Op}_h(r_0)=1+\textrm{Op}_h(c_0),
\]
where
\[
c_0=ht_1r_0+\sum_{|\alpha|=1}^4 \frac{h^{|\alpha|}}{i^{|\alpha|}\alpha!}\p_\xi^\alpha t\p_x^\alpha r_0\in hS^{2\delta,-1}_\delta.
\]
Next we shall determine $r_1\in S^{3\delta,-5}_\delta$ so that
\begin{equation}
\label{eq_r_1}
T\textrm{Op}_h(r_0+hr_1)=1+\textrm{Op}_h(c_1), \quad c_1\in h^{2}S^{4\delta,-2}_\delta.
\end{equation}
Arguing as above, we see that it suffices to choose $r_1$ so that $c_0+ht_0r_1=0$. With this choice, we get
\eqref{eq_r_1} with
\[
c_1=h^2t_1r_1+h\sum_{|\alpha|=1}^4\frac{h^{|\alpha|}}{i^{|\alpha|}\alpha!}\p_\xi^\alpha t\p_x^{\alpha} r_1\in h^{2}S^{4\delta,-2}_\delta.
\]
Iterating the above procedure with the choice $r_j=-h^{-j}t_0^{-1}c_{j-1}$ at each step, we get the result.

\end{proof}

Set
\[
r^{(N)}=\sum_{j=0}^N h^jr_j\in S^{\delta,-4}_\delta,
\]
where $N$ is large enough but fixed.
The operator $\textrm{Op}_h(r^{(N)})$ will serve as a right parametrix for our boundary value problem in the interior of $\overline{\Omega}$.

\subsection{The boundary parametrix}
Recall that we consider the region of the complex spectral plane, where  $\Im z\ge h^{\delta/2}$, $\delta>0$ small enough.
When constructing the parametrix for  \eqref{eq_bvp} near a boundary point, it will be convenient to straighten out the boundary locally by means of the boundary normal coordinates.
Let $x_0\in \p \Omega$ and introduce the boundary normal coordinates $y=(y',y_n)\in \textrm{neigh}(0,\R^n)$, $y'=(y_1,\dots,y_{n-1})$, centered at $x_0$.
Here $\textrm{neigh}(0,\R^n)$ stands for some open neighborhood of $0$ in $\R^n$.
In terms of $y$, locally near $x_0$,  $\p \Omega$ is defined by $y_n=0$, and $y_n>0$ if and only if $x\in \Omega$.
The principal symbol of $P_0$ expressed in the new coordinates becomes
\begin{equation}
\label{eq_p_0_1}
p_0(y,\eta)=\eta_n^2+s(y,\eta').
\end{equation}
Here $s(y',0,\eta')>0$ is the principal symbol of the Laplace--Beltrami  operator $-\Delta_{\p \Omega}$ on $\p \Omega$, expressed in the local coordinates $y'$,  see \cite{LeeUhl89}.

The problem \eqref{eq_bvp} in terms of the coordinates $y$ is given by
\begin{equation}
\label{eq_bvp_upper_half}
\begin{aligned}
T(y,hD_y, z;h) u(y',y_n)=f(y',y_n) \quad &\text{in}\quad \R^n_+,\\
u(y',y_n)|_{y_n=0}=g_1(y'), \quad &\\
 hD_{y_n} u(y',y_n)|_{y_n=0}=g_2(y'), \quad &
\end{aligned}
\end{equation}
where $\R^n_+$ is the half-space $y_n>0$.

Working locally near $y=0$ in  $\R^n$, let $\tilde f$ be the zero extension of $f$ to $\R^n$.
We shall look for the right parametrix of \eqref{eq_bvp_upper_half} in the form
\[
R(f,g_1,g_2)=\textrm{Op}_h(r^{(N)})(\tilde f)+R_b(\psi_1,\psi_2),
\]
where
\[
\psi_j(y')=g_j(y')-\gamma_0(hD_{y_n})^{j-1}\textrm{Op}_h(r^{(N)})(\tilde f), \quad j=1,2,
\]
and
\[
\gamma_0:u\mapsto u|_{y_n=0}
\]
is the restriction operator from the half space $\R^n_+$. Here $R_b$ should be a right parametrix of  the  boundary value problem
\begin{equation}
\label{eq_bvp_upper_half_0}
\begin{aligned}
T(y,hD_y, z;h) u(y',y_n)=0 \quad &\text{in}\quad \R^n_+,\\
u(y',y_n)|_{y_n=0}=\psi_1(y'), \quad &\\
 hD_{y_n} u(y',y_n)|_{y_n=0}=\psi_2(y'). \quad &
\end{aligned}
\end{equation}
We now shall construct $R_b$. In what follows,  we shall write  $(x',x_n)$ instead of $(y',y_n)$, and $(\xi',\xi_n)$ instead of $(\eta',\eta_n)$.

The construction will proceed similarly to \cite{SjoZwo93} and is essentially well-known in the theory of elliptic  boundary value problems, see e.g. \cite{ChazPir_book, Eskin_book, Eskin_notes, Grubbbook2009}.
For the convenience of the reader, we shall sketch a direct argument in the present semiclassical framework.

It follows from \eqref{eq_t_0_factor} together with \eqref{eq_p_0_1} that the equation
\[
t_0(x,\xi',\xi_n,z)=0
\]
has the solutions
\[
\xi_n=\sigma_j^+(x,\xi',z),\quad j=1,2,
\]
in the open upper half-plane, and the solutions
\[
\xi_n=\sigma_j^-(x,\xi',z),\quad j=1,2,
\]
in the open lower half-plane.
We have explicitly,
\begin{equation}
\label{eq_roots_1}
\sigma_1^\pm(x,\xi',z)=\pm \sqrt{z-s(x,\xi')}, \quad
\sigma_2^\pm(x,\xi',z)=\pm \sqrt{\frac{q+1}{q}z-s(x,\xi')},
\end{equation}
where we fix the branch of the square root with a positive imaginary part. In particular, we see that $\sigma_1^+(x,\xi',z)\ne \sigma_2^+(x,\xi',z)$ for all values of $x,\xi',z$.

For large $|\xi'|$, we have
$|\sigma_j^\pm(x,\xi',z)|\sim |\xi'|$ and
$|\Im \sigma_j^\pm(x,\xi',z)|\sim |\xi'|$, $j=1,2$. Furthermore,  $\sigma_j^\pm(x,\xi',z)\in S^1$ for large $|\xi'|$. Using that $\Im z\ge h^{\delta/2}$, we see that for bounded $|\xi'|$,
\[
|\Im \sigma_j^\pm(x,\xi',z)|\ge \frac{h^{\delta/2}}{C}.
\]

We have the factorization
\[
t_0(x,\xi',\xi_n)=q(x)t_0^+(x,\xi',\xi_n)t_0^-(x,\xi',\xi_n),
\]
\[
t_0^\pm(x,\xi',\xi_n)=(\xi_n-\sigma_1^\pm(x,\xi',z))(\xi_n-\sigma_2^\pm(x,\xi',z)).
\]

Recall from Proposition \ref{prop_interior_parametrix}  that $r^{(N)}(x,\xi',\xi_n,z;h)$ extends to a meromorphic function of $\xi_n\in \C$ with the poles at $\sigma_j^\pm(x,\xi',z)$. To be precise, following \cite {SjoZwo93}, let us notice that the function $r^{(N)}(x,\xi',\xi_n,z;h)$ belongs to the symbol class $S^{\delta, -4}_\delta$ in the domain
\[
\{(x,\xi',\xi_n): x\in\textrm{neigh}(0,\R^n), \xi'\in\R^{n-1}, \xi_n\in
\Omega(x,\xi',z)\},
\]
where
for large $|\xi'|$,
\[
\Omega(x,\xi',z)=\{\xi_n\in \C:|\xi_n|\le C\langle \xi'\rangle,|\xi_n-\sigma_j^+|\ge \langle \xi'\rangle/C,j=1,2,\Im \xi_n\ge \frac{1}{C}\langle \xi'\rangle\}
\]
whereas, for $|\xi'|=\mathcal{O}(1)$,
\[
\Omega(x,\xi',z)=\Omega_1(x,\xi',z)\cup\Omega_2(x,\xi',z),
\]
\[
\Omega_1(x,\xi',z)=\{|\xi_n|\le C, 0\le \Im \xi_n\le \frac{h^{\delta/2}}{C}\}
\]
and
\[
\Omega_2(x,\xi',z)=
\{|\xi_n|\le C, \Im \xi_n\ge \frac{h^{\delta/2}}{C},|\xi_n-\sigma_j^+(x,\xi',z)|\ge \frac{1}{C},j=1,2\}.
\]
Here $C>0$ is an arbitrarily large  but fixed constant. This follows from our estimates for the roots $\sigma_j^\pm$.

Let $\gamma=\gamma(x,\xi',z;h)$ be a simple closed $C^1$ curve in $\Omega(x,\xi',z)$, which encircles the roots $\sigma_1^+(x,\xi',z)$ and $\sigma_2^+(x,\xi',z)$
in the positive sense, and such that the length of $\gamma$ is $\mathcal{O}(\langle\xi' \rangle)$.

Continuing to follow  \cite{SjoZwo93}, locally near $0$,
we define the operators,
\begin{align*}
&\Pi_j:C^\infty_0(\textrm{neigh}(0,\R^{n-1}))\to C^\infty(\textrm{neigh}(0,\overline{\R}^n_+)), \\
&\Pi_j\varphi(x)=\frac{1}{(2\pi h)^{n-1}}\int_{\xi'\in\R^{n-1}}\int_{\xi_n\in \gamma} e^{ix\cdot\xi/h}r^{(N)}(x,\xi,z;h)\xi_n^j\hat{\varphi}\bigg(\frac{\xi'}{h}\bigg)d\xi'\frac{1}{2\pi i}d\xi_n,
\end{align*}
$ j=0,1$.
From \cite{ChazPir_book}, we recall the following mapping properties,
\begin{equation}
\label{eq_mapping_prop}
\Pi_j:H^s_0(\textrm{neigh}(0,\R^{n-1}))\to H^{s+4-j-1/2}(\textrm{neigh}(0,\R^n_+)),\quad s\in \R.
\end{equation}

As  the poles of the meromorphic function $\xi_n\mapsto r^{(N)}(x,\xi',\xi_n,z;h)$ in the upper half-plane  are precisely $\sigma_j^+$, $j=1,2$, a contour deformation argument in the $\xi_n$-plane shows that
\begin{equation}
\label{eq_pi_r_+}
\Pi_j\varphi=\frac{h}{i} \textrm{Op}_h(r^{(N)})(\varphi\otimes (hD_{x_n})^j\delta_{x_n=0}), \quad j=1,2, \quad x_n>0.
\end{equation}
The operators $\Pi_j$ can therefore be viewed as Poisson operators for the boundary value problem \eqref{eq_bvp_upper_half_0}.
Using that the operator $T$ is local together with   \eqref{eq_pi_r_+},
we get from Proposition \ref{prop_interior_parametrix},
\begin{equation}
\label{eq_TPI}
T\Pi_j\varphi=h^{N+1}\textrm{Op}_h(c_N)(\varphi\otimes (hD_{x_n})^j\delta_{x_n=0}),\quad x_n>0,
\end{equation}
with $c_N\in S^{2\delta(N+1),-N-1}_\delta$.

We shall construct the parametrix $R_b$ of the boundary value problem \eqref{eq_bvp_upper_half_0} in the form,
\[
R_b(\psi_1,\psi_2)=\Pi_0(\varphi_0)+\Pi_1(\varphi_1)
\]
for some functions $\varphi_0, \varphi_1$, defined locally near $0\in \R^{n-1}$, to be determined. In view of \eqref{eq_TPI},
we need only to compute $\gamma_0\Pi_j$ and $\gamma_0hD_{x_n}\Pi_j$, $j=0,1$.

Let $r^{(N)}=r_0+h\tilde r^{(N)}$, where  $r_0\in S^{\delta, -4}_{\delta}$ and $\tilde r^{(N)}\in S^{3\delta, -5}_{\delta}$. Then  we have
\[
\gamma_0\Pi_j=\textrm{Op}_h(d_j),
\]
where
\[
d_j=\frac{1}{2\pi i}\gamma_0\int_{\xi_n\in \gamma} e^{ix_n\xi_n/h}\xi_n^j r^{(N)}(x,\xi,z;h)d\xi_n=d_{j,0}+h\tilde d_j.
\]
with
\[
d_{j,0}=\frac{1}{2\pi i}\int_{\xi_n\in \gamma} \xi_n^j r_0(x',0,\xi,z)d\xi_n\in S^{\delta, -4+j+1}_\delta,
\]
and $\tilde d_j\in S^{3\delta, -4+j}_\delta$. Here we have used the assumption that the length of  the contour $\gamma$ is $\mathcal{O}(\langle \xi'\rangle)$.

The residue calculus gives that
\begin{equation}
\label{eq_d}
d_{j,0}(x',\xi',z)=\sum_{\nu=1}^2\frac{(\sigma_\nu^+)^j}{\p_{\xi_n}t_0(x',0,\xi',\sigma_\nu^+)}.
\end{equation}

For $j=0,1$, we compute next
\begin{align*}
\gamma_0hD_{x_n}\Pi_j\varphi=\frac{1}{(2\pi h)^{n-1}}\int_{\xi'\in\R^{n-1}}e^{ix'\cdot \xi'/h}c_j(x',\xi',z')
\hat{\varphi}\bigg(\frac{\xi'}{h}\bigg)d\xi'=\textrm{Op}_h(c_j)\varphi,
\end{align*}
where
\begin{align*}
c_j&=
\frac{1}{2\pi i}\gamma_0\int_{\xi_n\in \gamma} (\xi_n^{j+1}r^{(N)}(x,\xi,z;h)+ \xi_n^j hD_{x_n}r^{(N)}(x,\xi,z;h))e^{ix_n\xi_n/h} d\xi_n\\
&=c_{j,0}+h\tilde c_{j}.
\end{align*}
Here
\[
c_{j,0}=\frac{1}{2\pi i}\int_{\xi_n\in \gamma} \xi_n^{j+1}r_0(x',0,\xi',\xi_n,z)d\xi_n\in S_\delta^{\delta,-2+j},
\]
and $\tilde c_j\in S^{3\delta,-3+j}_\delta$.
We have
\begin{equation}
\label{eq_c}
c_{j,0}=\sum_{\nu=1}^2\frac{(\sigma_\nu^+)^{j+1}}{\p_{\xi_n}t_0(x',0,\xi',\sigma_\nu^+)}.
\end{equation}

Hence, we obtain the following pseudodifferential system on the boundary,
\[
A
\begin{pmatrix}\varphi_0\\
\varphi_1
\end{pmatrix}= \begin{pmatrix} \psi_1\\
\psi_2
\end{pmatrix}, \quad A=\begin{pmatrix} \gamma_0\Pi_0& \gamma_0\Pi_1\\
\gamma_0hD_{x_n}\Pi_0& \gamma_0h D_{x_n}\Pi_1
\end{pmatrix}
=
\begin{pmatrix}
\textrm{Op}_h(d_0) & \textrm{Op}_h(d_1)\\
\textrm{Op}_h(c_0) & \textrm{Op}_h(c_1)
\end{pmatrix}.
\]
In view of \eqref{eq_d} and \eqref{eq_c}, we see that the semiclassical principal symbol of $A$  is given by
\begin{align*}
a(x',\xi',z)&=\begin{pmatrix}
\frac{1}{\p_{\xi_n}t_0(\sigma_1^+)}+\frac{1}{\p_{\xi_n}t_0(\sigma_2^+)} & \frac{\sigma_1^+}{\p_{\xi_n}t_0(\sigma_1^+)}+\frac{\sigma_2^+}{\p_{\xi_n}t_0(\sigma_2^+)} \\
\frac{\sigma_1^+}{\p_{\xi_n}t_0(\sigma_1^+)}+\frac{\sigma_2^+}{\p_{\xi_n}t_0(\sigma_2^+)} & \frac{(\sigma_1^+)^2}{\p_{\xi_n}t_0(\sigma_1^+)}+\frac{(\sigma_2^+)^2}{\p_{\xi_n}t_0(\sigma_2^+)}
\end{pmatrix}\\
&=
\begin{pmatrix}
1 & 1\\
\sigma_1^+ & \sigma_2^+
\end{pmatrix}
\begin{pmatrix}
\frac{1}{\p_{\xi_n}t_0(\sigma_1^+)} & 0\\
0 & \frac{1}{\p_{\xi_n}t_0(\sigma_2^+)}
\end{pmatrix}
\begin{pmatrix}
1 & \sigma_1^+ \\
1 & \sigma_2^+
\end{pmatrix}.
\end{align*}
Writing $a=(a_{jk})$, $1\le j,k\le 2$, we observe that $a_{jk}\in S_\delta^{\delta,j+k-5}$. In order to invert $A$, let us consider $\det(a(x',\xi',z))\in S^{2\delta, -4}_\delta$. It follows from \eqref{eq_roots_1} that for large $|\xi'|$,
\[
|\det (a(x',\xi',z)) |\sim \langle \xi'\rangle^{-4},
\]
while for $|\xi'|=\mathcal{O}(1)$,
\[
|\det (a(x',\xi',z)) |\ge 1/C.
\]
The Fa\`a di Bruno formula \cite{Lerner_book} implies that
for large $|\xi'|$,
\[
\bigg|\p_{x'}^\alpha\p_{\xi'}^\beta \frac{1}{\det (a)}\bigg|\le C_{\alpha,\beta}\langle \xi'\rangle^{4-|\beta|},
\]
and for $|\xi'|=\mathcal{O}(1)$,
\[
\bigg|\p_{x'}^\alpha\p_{\xi'}^\beta \frac{1}{\det (a)}\bigg|\le C_{\alpha,\beta}h^{-3\delta(|\alpha|+|\beta|)}.
\]
Hence, $1/\det(a)\in S^{0,4}_{3\delta}$. It follows that if $b=a^{-1}=(b_{jk})$, $1\le j,k\le 2$, that $b_{jk}\in S^{\delta,-j-k+5}_{3\delta}$.
We obtain that
\[
A\textrm{Op}_h(b)=I-h\textrm{Op}_h(e),\quad e\in
\begin{pmatrix}
S^{4\delta,-1}_{3\delta} & S^{4\delta,-2}_{3\delta}\\
S^{4\delta,0}_{3\delta} & S^{4\delta,-1}_{3\delta}
\end{pmatrix},
\]
provided that $\delta<1/6$. Let
\begin{equation}
\label{eq_mapping_B}
B^{(N)}=\textrm{Op}_h(b)\sum_{k=0}^{N-1}(h\textrm{Op}_h(e))^k\in \textrm{Op}_h\begin{pmatrix}
S^{\delta,3}_{3\delta} & S^{\delta,2}_{3\delta}\\
S^{\delta,2}_{3\delta} & S^{\delta,1}_{3\delta}
\end{pmatrix},
 \quad N\in \N.
\end{equation}
Then
\[
AB^{(N)}=I-h^N\textrm{Op}_h(e^{(N)}), \quad e^{(N)}
\in
\begin{pmatrix}
S^{4\delta N,-N}_{3\delta} & S^{4\delta N,-N-1}_{3\delta}\\
S^{4\delta N,-N+1}_{3\delta} & S^{4\delta N,-N}_{3\delta}
\end{pmatrix}.
\]
Introducing
\begin{equation}
\label{eq_G}
G_0=\Pi_0B_{11}^{(N)}+\Pi_1 B_{21}^{(N)},\quad G_1=\Pi_0B_{12}^{(N)}+\Pi_1 B_{22}^{(N)},
\end{equation}
we define the boundary parametrix $R_b$ by
\[
R_b(\psi_1,\psi_2)=G_0\psi_1+G_1\psi_2.
\]
Thus, we have
\begin{equation}
\label{eq_1}
\begin{aligned}
\gamma_0 R_b(\psi_1,\psi_2)=&\psi_1-h^N \textrm{Op}_h(e_{11}^{(N)})\psi_1-h^N \textrm{Op}_h(e_{12}^{(N)})\psi_2,\\
\gamma_0 hD_{x_n} R_b(\psi_1,\psi_2)=&\psi_2-h^N \textrm{Op}_h(e_{21}^{(N)})\psi_1-h^N \textrm{Op}_h(e_{22}^{(N)})\psi_2.
\end{aligned}
\end{equation}
Also, the kernel of the operator $TG_{j}$, $j=0,1$, satisfies
\begin{equation}
\label{eq_kernel}
|\p_x^\alpha\p_{y'}^\beta TG_j(x,y', z;h)|\le \mathcal{O}(h^M),\quad |\alpha|+|\beta|\le M,
\end{equation}
where
$M=M(N)\to \infty$, as $N\to \infty$. When verifying \eqref{eq_kernel}, it suffices to consider $T\Pi_0B^{(N)}_{11}$, since the treatment of the other terms in
\eqref{eq_G}
 is similar.
It follows from \eqref{eq_TPI} that
\[
T\Pi_0B^{(N)}_{11}\varphi=h^{N+1}\textrm{Op}_h(c_N)(B^{(N)}_{11}\varphi\otimes\delta_{x_n=0}),\quad x_n>0,
\]
with $c_N\in S^{2\delta(N+1),-N-1}_\delta$. The kernel of this operator is of the form
\[
\frac{h^{N+1}}{(2\pi h)^{2n-1}}\int e^{ix_n\xi_n/h}e^{i(x'-y')\cdot\xi'/h}e^{i(y'-z')\cdot\eta'/h}c_N(x,\xi)b_{11}^{(N)}(y',\eta')dy'd\eta' d\xi,
\]
and since $c_N\in S^{2\delta(N+1),-N-1}_\delta$, it is easy to see that \eqref{eq_kernel} holds.

The right parametrix of \eqref{eq_bvp_upper_half} takes the form
\begin{equation}
\label{eq_param_r_i+r_b}
R(f,g_1,g_2)=R_{\textrm{int}}(f) + R_b(g_1,g_2),
\end{equation}
where
\[
R_{\textrm{int}}(f)= \textrm{Op}_h(r^{(N)})(\tilde f) -G_0\gamma_0\textrm{Op}_h(r^{(N)})(\tilde f)-G_1\gamma_0 hD_{x_n}\textrm{Op}_h(r^{(N)})(\tilde f).
\]
One can also check that the kernels of
$TG_j\gamma_0(hD_{x_n})^j\textrm{Op}_h(r^{(N)})$, $j=0,1$, satisfy the  estimates
\begin{equation}
\label{eq_3}
|\p_x^\alpha\p_{y}^\beta TG_j \gamma_0 (hD_{x_n})^j\textrm{Op}_h(r^{(N)})
(x,y, z;h)|\le \mathcal{O}(h^M), |\alpha|+|\beta|\le M,  j=0,1,
\end{equation}
where
$M=M(N)\to \infty$, as $N\to \infty$. We refer to  \cite[Section 3]{SjoZwo93} for the details of this verification  based on a contour deformation argument in the complex $\xi_n$-plane and repeated integration by parts. Finally, we have
\begin{equation}
\label{eq_4}
|\p_{x'}^\alpha\p_{y}^\beta (\gamma_0 (hD_{x_n})^j R_{\textrm{int}})
(x',y, z;h)|\le \mathcal{O}(h^M), |\alpha|+|\beta|\le M,  j=0,1,
\end{equation}
where
$M=M(N)\to \infty$, as $N\to \infty$. This completes the construction of the right parametrix for the problem \eqref{eq_bvp_upper_half}.

\subsection{Global parametrix}
We can find finitely many points $x_j\in  \overline{\Omega}$, $1\le j\le L$, such that $x_j\in \Omega$, $1\le j\le L'$, $x_j\in \p \Omega$, $L'+1\le j\le L$, and neighborhoods $U_j$ of $x_j$ forming an open cover of $\overline{\Omega}$ such that we can introduce  boundary normal coordinates in each $U_j$, $L'+1\le j\le L$.   Let $\varphi_j\in C^\infty_0(U_j)$ form a partition of unity in $\overline\Omega$.
Take $\psi_j\in C^\infty_0(U_j)$ with $\psi_j=1$ near $\supp(\varphi_j)$. Then define the global parametrix $\mathcal{R}=\mathcal{R}(z;h)$
by
\[
\mathcal{R}(f,g_1,g_2)=\sum_{j=1}^{L'}\psi_j\textrm{Op}_h(r^{(N)})\varphi_j f +\sum_{j=L'+1}^L\psi_j R_j (\varphi_j  f,\varphi_j|_{\p \Omega} g_1, \varphi_j|_{\p \Omega} g_2).
\]
Here when $L'+1\le j\le L$, in the boundary normal coordinates in $U_j$,  $R_j$ is of the form \eqref{eq_param_r_i+r_b}.

Let us recall the space $\mathcal{H}^s$ introduced in \eqref{eq_space_H_s}.  As before, we  equip the spaces $\mathcal{H}^s$ and $H^s(\Omega)$ with the natural semiclassical norms. Then it follows from \eqref{eq_mapping_prop}, \eqref{eq_mapping_B} and \eqref{eq_G} that
 the operator
\[
\mathcal{R}:\mathcal{H}^{4}\to H^4(\Omega),
\]
is bounded,  where we do not insist on any uniformity with respect to $h$.

Recall from Section 3 the operator
$\mathcal{T}= (T,\gamma_0,\gamma_0 hD_\nu)$.
Then it is  standard to see, using Proposition \ref{prop_interior_parametrix} together with \eqref{eq_1}, \eqref{eq_kernel}, \eqref{eq_3},  \eqref{eq_4},  that the operator $\mathcal{R}$ satisfies,
\[
\mathcal{T}\mathcal{R}-1=\mathcal{O}(h^M):\mathcal{H}^4\to \mathcal{H}^4,
\]
with
$ M=M(N)\to \infty$, as  $N\to \infty$.
We conclude that for $h$ small enough, the operator $\mathcal{T}$  has the right inverse.
By Proposition \ref{prop_index_0}, it follows that the operator $\mathcal{T}: H^4(\Omega)\to \mathcal{H}^4$ is invertible for $h$ small enough and $\Im z\ge h^{\delta/2}$, for $\delta >0$ sufficiently small. Applying the semiclassical reduction of Section 3, we
obtain that there exists a constant $C>0$ such that all transmission eigenvalues $\lambda\in \C$ with $|\lambda|\ge C$ satisfy $|\Im \lambda|\le C |\lambda|^{1-\delta/4}$.

\section{The region $\Re \lambda\le 0$ of the complex plane}

\label{sec_left_half}

In order to complete the proof of Theorem \ref{thm_main}, it remains to show that the left half-plane $\Re \lambda\le 0$ contains at most finitely many transmission eigenvalues.

Let $(\cdot,\cdot)$ be the scalar product in $L^2(\Omega)$. When $u\in H^4(\Omega)\cap H^2_0(\Omega)$ and $\Re \lambda<0$, we have
\begin{align*}
\Re (T(\lambda)u,u)&=(qP_0u,P_0u)+2|\Re \lambda| \Re (P_0 u,q u)+|\Re \lambda|(P_0 u,u)\\
&+((\Re \lambda)^2-(\Im \lambda)^2)((1+q)u,u).
\end{align*}
We have already established  that all but finitely many transmission eigenvalues belong to the region
\[
|\Re \lambda|\ge C|\Im \lambda|,
\]
where $C>0$ is the constant that can be taken arbitrarily large. Restricting the attention to this region, for $C$ large enough, we get
\begin{align*}
\Re (T(\lambda)u,u)&\ge 2|\Re \lambda| \Re (P_0 u,q u)+|\Re \lambda|(P_0 u,u)+ \frac{1}{2}(\Re \lambda)^2\|u\|^2,\\
& = 2|\Re \lambda|\|q^{1/2}\nabla u\|^2+2|\Re \lambda| \Re (\nabla u, u\nabla q)\\
&+
|\Re \lambda| \|\nabla u\|^2+ \frac{1}{2}(\Re \lambda)^2\|u\|^2.
\end{align*}
Using the inequality
\begin{align*}
2|\Re \lambda| |(\nabla u,u\nabla q)|\le 2 |\Re \lambda|\|u\nabla q\|\|\nabla u\|\le \varepsilon  |\Re \lambda|^2\|\nabla q\|_{L^\infty}^2\|u\|^2+\frac{1}{\varepsilon}\|\nabla u\|^2
\end{align*}
with $\varepsilon>0$, we  obtain that
\begin{align*}
\Re (T(\lambda)u,u)\ge \bigg(\frac{1}{2}-\varepsilon \|\nabla q\|_{L^\infty}^2\bigg)(\Re \lambda)^2 \|u\|^2+ \bigg(|\Re \lambda|-\frac{1}{\varepsilon}\bigg)\|\nabla u\|^2.
\end{align*}
Choosing $\varepsilon$ small enough, we see that the region $\Re \lambda <-2 \|\nabla q\|^2_{L^\infty}$ does not contain any transmission eigenvalues.
Since the strip $\Re \lambda\in [-2 \|\nabla q\|^2_{L^\infty},0]$ contains at most finitely many transmission eigenvalues, the result follows.
This completes the proof of Theorem \ref{thm_main}.


\begin{appendix}

\section{Basic facts on semiclassical calculus}

Let us start by recalling the definition of the following standard symbol class. Let $S^k(\R^n\times\R^n)$ be the space of symbols $a(x,\xi,z;h)$, which are $C^\infty$ with respect to $(x,\xi)\in \R^n\times \R^n$, such that for all $\alpha,\beta\in \N^n$, there is a constant $C_{\alpha,\beta}$ so that uniformly in $h$ and $z$, we have
\[
|\p ^\alpha_x\partial^\beta_\xi a(x,\xi,z;h)|\le C_{\alpha,\beta}\langle \xi \rangle^{k-|\beta|}, \quad \textrm{for all}\quad (x,\xi)\in \R^n\times \R^n.
\]
Here $\langle \xi \rangle=(1+|\xi|^2)^{1/2}$.

The classical quantization of the symbol $a\in S^k$ is given by
\[
\textrm{Op}_h(a)u(x)=\frac{1}{(2\pi h)^n}\int_{\R^n} e^{i x\cdot\xi/h } a(x,\xi,z; h)\hat u(\xi/h)d\xi,
\]
where
\[
\hat{u}(\xi)=\int_{\R^{n}} e^{-ix\cdot\xi}u(x)dx
\]
is the Fourier transform.

Let $a\in S^{k_1}$ and $b\in S^{k_2}$. We have
\[
\textrm{Op}_h(a)\textrm{Op}_h(b)=\textrm{Op}_h(a\#b),
\]
where $a\#b\in  S^{k_1+k_2}$ is given by
\begin{equation}
\label{eq_comp_symbol}
\begin{aligned}
a\#b(x,\xi,z;h)&=e^{-ix\cdot \xi/h}\textrm{Op}_h(a)(b(\cdot,\xi)e^{i(\cdot)\cdot \xi/h})\\
&=\frac{1}{(2\pi h)^n}\int\!\!\! \int a(x,\eta,z;h)b(y,\xi,z;h)e^{\frac{i}{h}(x-y)\cdot (\eta-\xi)}dyd\eta.
\end{aligned}
\end{equation}
Moreover, the symbol $a\#b$ has the asymptotic expansion, see \cite{Sjostrand_book},
\begin{equation}
\label{eq_comp_expansion}
a\#b(x,\xi,z;h)\sim \sum_{|\alpha|\ge 0} \frac{h^{|\alpha|}}{\alpha!} \p_\xi^\alpha a D^\alpha_x b,
\end{equation}
in the sense that for any $N$,
\[
a\#b(x,\xi,z;h)- \sum_{|\alpha|< N} \frac{h^{|\alpha|}}{\alpha! }\p_\xi^\alpha a D^\alpha_x b\in h^{N}S^{k_1+k_2-N}.
\]

When studying the invertibility of the family $T$ given by \eqref{eq_family_T}, we shall encounter symbols having a slightly degenerate behavior in the region where $|\xi|$ is bounded, and to keep track of that we introduce the following symbol class,  based on $S^k(\R^n\times\R^n)$.
When $0\le \delta<1/2$, $m\ge 0$, $k\in\R$, we let $S^{m,k}_{\delta}(\R^n\times\R^n)$ stand for the space of symbols $a(x,\xi,z;h)$, which are $C^\infty$ with respect to $(x,\xi)\in \R^n\times \R^n$, such that
\begin{itemize}
\item[(1)] $\supp(a)\subset K\times \R^n$, where $K$ is a compact subset of $\R^n$,
\item[(2)]  for $\xi$ outside some $h$-independent compact  set,  $a\in S^k$,
\item[(3)]  $a$ satisfies
\[
|\p ^\alpha_x\partial^\beta_\xi a(x,\xi,z;h)|\le C_{\alpha,\beta,L}h^{-m}h^{-\delta(|\alpha|+|\beta|)}, \quad \textrm{for all}\quad (x,\xi)\in \R^n\times L,
\]
for all $\alpha,\beta\in \N^n$ and all compact sets $L\subset \R^n$.
\end{itemize}

Let $a\in S^{m_1,k_1}_\delta$ and $b\in S^{m_2,k_2}_\delta$. Then we shall show that  the symbol $a\#b\in S^{m_1+m_2,k_1+k_2}_\delta$.
Indeed, let $\chi\in C_0^\infty(\R^n)$, $\supp(\chi)\subset \{\xi:|\xi|<1/2\}$ and $\chi=1$ near $0$. Then using \eqref{eq_comp_symbol} let us write
$a\#b=c_1+c_2$, where
\begin{align*}
c_2(x,\xi,&z;h)\\
&=\frac{1}{(2\pi h)^n}\int\!\!\! \int a(x,\eta,z;h)b(y,\xi,z;h)\bigg(1-\chi\bigg(\frac{\eta-\xi}{\langle \xi\rangle}\bigg)\bigg)e^{\frac{i}{h}(x-y)\cdot (\eta-\xi)}dyd\eta.
\end{align*}
Carrying out repeating partial integrations with respect to the variable $y$, we see that for any $N$, $\alpha$, $\beta$, there exits $C_{N\alpha\beta}>0$ such that
\[
|\p_x^\alpha\p_\xi^\beta c_2(x,\xi,z;h)|\le C_{N\alpha\beta}h^N\langle \xi\rangle ^{-N},\quad (x,\xi)\in \R^n\times\R^n.
\]
Now
 \begin{align*}
c_1(x,\xi,z;h)
=\frac{1}{(2\pi h)^n}\int\!\!\! \int a(x,\eta,z;h)b(y,\xi,z;h)\chi\bigg(\frac{\eta-\xi}{\langle \xi\rangle}\bigg)e^{\frac{i}{h}(x-y)\cdot (\eta-\xi)}dyd\eta\\
=\bigg(\frac{\langle \xi\rangle}{2\pi h}\bigg)^n\int\!\!\! \int a(x,\langle\xi\rangle\eta+\xi,z;h)b(x+y,\xi,z;h)\chi(\eta)e^{-\frac{i\langle\xi\rangle}{h}y\cdot \eta}dyd\eta.
\end{align*}
It follows from the standard semiclassical calculus \cite{Sjostrand_book},  that $c_1\in S^{m_1+m_2,k_1+k_2}_\delta$, with the natural asymptotic expansion, similar to \eqref{eq_comp_expansion}.
In particular,
\[
a\#b-ab\in hS^{m_1+m_2+2\delta, k_1+k_2-1}_\delta.
\]

Finally let us recall the following mapping properties of the classical quantization of the symbol $a\in S^{m,k}_{\delta}$, $0\le \delta<1/2$, $m\ge 0$, $k\in \R$, see \cite{Sjostrand_book},
\[
\textrm{Op}_h(a):H^s(\R^n)\to H^{s-k}(\R^n),\quad s\in \R,
\]
and
\[
\|\textrm{Op}_h(a)\|_{H^s\to H^{s-k}}\le \mathcal{O}(h^{-m}).
\]
Here the standard Sobolev space $H^s(\R^n)$ has been equipped with the natural semiclassical norm
\[
\|u\|_{H^s}^2=\frac{1}{(2\pi)^n}\int_{\R^n} (1+|h\xi|^2)^s|\hat u(\xi)|^2 d\xi.
\]

\end{appendix}

\section*{Acknowledgements}
 The research of M.H. was partially supported by the NSF grant DMS-0653275 and he is grateful to the Department of Mathematics and Statistics at the University of Helsinki for the hospitality. The research of K.K. was financially supported by the
Academy of Finland (project 125599).
The research of P.O. and L.P. was financially supported by Academy of Finland Center of Excellence programme 213476.

\end{document}